\documentclass[a4paper, 12pt]{article}
\usepackage{pdflscape}
\usepackage[utf8]{inputenc}
\usepackage[russian]{babel}
\usepackage{geometry}
\usepackage{amsthm,amsmath,amscd}
\usepackage{mathrsfs}
\usepackage{relsize}
\usepackage{amsfonts,amssymb}
\usepackage{mathtools}
\usepackage{titling}

\newlength{\bigtextsize}
\usepackage{soulutf8}
\usepackage{icomma}
\usepackage{hyperref}
\usepackage{graphicx}
\usepackage{enumitem}

\newtheorem{theorem}{Теорема}[section]

\newtheorem{lemma}{Лемма}[section]

\newtheorem{remark}{Замечание}[section]
\newtheorem{statement}{Утверждение}[section]
\newtheorem{definition}{Определение}[section]

\newcommand{\zs}{\ensuremath{\mathbb{Z}}}

\begin{document}

\title{\textbf{\large Вероятностные методы обхода лабиринта с использованием камней и датчика случайных чисел}}
\author{Е.Г. Кондакова, МФТИ, А. Я. Канель-Белов, SZU }
\date{}
\maketitle
\begin{abstract}

Существует широкий спектр задач посвященных возможности обхода лабиринта конечными автоматоми.
Они могут отличаться как типом лабиринта(это может быть любой граф, даже бесконечный), так и самими автоматами или их количеством.
В частности у конечного автомата может быть память(магазин) или генератор случайных битов.
В дальнейшем будем считать, что робот --- это конечный автомат с генератором случайных битов, если не сказано иное.
Кроме того в этой системе могут быть камни-объект, который конечный автомат может переносить по графу, и флажки- объект, наличие которого конечный автомат может только "наблюдать".
Эта тема представляет интерес в связм с тем, что некоторые из этих задач тесно связаны с задачами из теории вероятности и сложности вычислений.

В данной работе продолжают решаться некоторые открытые вопросы, поставленные в диссертации Аджанса: обход роботом с генератором случайных битов целочисленных пространств при наличии камня и подпространства флажков~\cite{And1}.
	Подобные задачи помогают развить математический аппарат в данной области, кроме того в этой работе мы исследуем практически не изученное поведение робота с генератором случайных чисел.
	Представляется чрезвычайно важным перенос комбинаторных методов, разработанных А. М. Райгородским в задачах этой тематики.

Данная работа посвящена обходу лабиринта конечным автоматом с генератором случайных битов.
Эта задача является частью активно развивающейся темы обхода лабиринта различными конечными автоматами
или их коллективами, которая тесно связана с задачами из теории сложности вычислений и теории вероятности.
В данной работе показано, при каких размерностях робот с генератором случайных битов и камнем может обойти
целочисленное пространство с подпостранством флажков. В данной работе будет изучено поведение конечного автомата с генератором случайных битов на целочисленных пространствах.
В частности доказано, что
робот обходит $\zs^2$ и не может обойти $\zs^3$;
робот c камнем обходит $\zs^4$ и не может обойти $\zs^5$;
робот c камнем и флажком обходит $\zs^6$ и не может обойти $\zs^7$;
робот c камнем и плоскостью флажков обходит $\zs^8$ и не может обойти $\zs^9$.

Работа поддержана Российским Научным Фондом, грант №17-11-01377.
\end{abstract}
\vspace{1cm}

\section*{Введение}
	Данная работа связана с задачей о поведении роботов в лабиринтах.
	Вопросы этой темы представляют большой интерес.
	Это связано с тем, что продвижения в некоторых важных задачах теоретической Computer Sсince могут быть получены из области поведения роботов в лабиринтах.
	Такое положение дел делает актуальными задачи данной тематики.
	Роботы в лабиринтах очень важны, в частности, им посвящено много литературы~\cite{Kilib1,Kilib3,Kilib4}.
	В данной работе продолжают решаться некоторые открытые вопросы, поставленные в диссертации Аджанса: обход роботом с генератором случайных битов целочисленных пространств при наличии камня и подпространства флажков~\cite{And1}.
	Подобные задачи помогают развить математический аппарат в данной области, кроме того в этой работе мы исследуем практически не изученное поведение робота с генератором случайных чисел.
	Представляется чрезвычайно важным перенос комбинаторных методов, разработанных А. М. Райгородским в задачах этой тематики.
	
	Задачи этой темы могут принимать разные виды, но есть несколько общих
	частей.
	 Основным элементом является робот или коллектив роботов~\cite{Kilib5}.
	   Робот --- это некоторый конечный автомат, но у него может быть генератор случайных битов или память в какой-то форме.
	Он перемещается в некоторой среде~\cite{Kilib2}.
	Её называют лабиринтом.
	В этой среде могут быть камни, которые робот может переносить, кроме того эта среда может быть раскрашена.
	Множество флагов в лабиринте, которые робот может увидеть, но
	не может переносить, по сути является двухцветной раскраской. В дальнейшем будем рассматривать n-цветные раскраски пространств как (n-1)-цветную раскраску подмножества флажков.
	Робот решает разные задачи такие как: встречи двух роботов, распознавание типа лабиринта и его обход, что фактически является решением нахождения клетки выхода из него.

\subsection*{Имеющиеся результаты для робота}
	В задачах ниже робот считается обходящим пространство, если для любой клетки пространства вероятность в ней побывать равна единице.
	
	В моей бакалаврской работе были доказаны простейшие случаи возможности обхода роботом с генератором случайных чисел целочисленного пространства:
	\begin{itemize}
		\item
		Робот обходит $\zs^2$ и не может обойти $\zs^3$,
		\item
		Робот c камнем обходит $\zs^4$ и не может обойти $\zs^5$,
		\item
		Робот c камнем и флажком обходит $\zs^6$ и не может обойти $\zs^7$,
		\item
		Робот c камнем и плоскостью флажков обходит $\zs^8$ и не может обойти $\zs^9$.
	\end{itemize}
	
	Аналогично решенным задачам можно показать, что робот c камнем и подпространством флажков размерности $\zs^n$ не может обойти $\zs^n+7$
	
	Увеличение количества камней также не представляет интереса, так как робот с двумя камнями может работать, как машина Минского (это будет показано в данной работе). А машина Минского может обойти любое $\zs^n$.

\section{Базовые определения нашей модели} \label{chapt1}

\begin{subsection}{Базовые определения системы робот-лабиринт}
Сначала дадим основные определения, которые нам понадобятся для понимания теоретической основы разбираемой темы. В рамках данной работы средой, в которой движется робот, будем считать некоторую группу. Тогда нашу модель можно описать с помощью следующих определений. Часть из этих терминов даны в урезанном варианте и рассматривается только с точки зрения решаемых задач.

\begin{definition}
{\em Системой лабиринт-робот} называется $L = (G,R,n,D,M)$, где $G$ --- некоторая группа, $R$ --- конечный автомат специального вида, $n$ --- неотрицательное целое число, $D$ --- подмножество $G$, $M$ --- подмножество $G$. (\cite{Kilib6}).
\end{definition}

\begin{definition}
Конечный автомат с генератором случайных битов $R=(Q,q_{0},\delta,\xi=(\xi _{1},\dots,\xi _{k}, \dots))$ называется {\it автоматом робота}, где $Q$ --- множество состояний автомата, $q_{0}\in Q$ --- начальное состояние автомата, $\delta \subset Q\times {0,1}^{(n+2)}\rightarrow Q$ --- функция переходов в автомате, $\xi =(\xi_{1},\dots,\xi_{k}, \dots)$ ---последовательность независимых одинаково распределенных случайных величин, $P(\xi_{i}=0)=P(\xi_{i}=1)=\frac{1}{2}$.
\end{definition}

\begin{remark}
Переходы в автомат $R=(Q,q_{0},\delta,\xi)$ осуществляются по битовым векторам длины $(n+2)$
\end{remark}

\begin{definition}
Состоянием системы лабиринт-робот называется набор $(a,s_{1},\dots,s_{n},k,q)$, где $a\in Q$, $s _{i}\in Q$, $q\in Q$, $k\in \mathbb{N}$. Будем называть $a$ --- расположением робота, $s _{i}$ --- расположениями камней, $k$ --- номером хода робота, соответствующим случайной величине $\xi _{k}$,  которая дает случайный бит. Начальным состоянием системы лабиринт-робот будет $(e,e,e,\dots,e,1,q _{0})$, где $e$ --- нейтральный элемент группы G.
\end{definition}

\begin{definition}
$(d_1,d_2,\dots,d_m) = D$ называются {\em переходными элементами} системы лабиринт-робот.
\end{definition}

\begin{definition}
Ходом в состоянии $q\in Q$ называется пара (d,p), где $d\in D$, $p\in {\{0,1\}}^{n}$.
\end{definition}

По сути ход соответствует перемещению робота в $G$ и множеству камней, которые он перенесёт с собой. Каждому элементу $Q\times {\{0,1\}}^{n+2}$ соответствует свой ход.

\begin{definition}
{\em Результатом хода робота} в состоянии системы $(a, s_l,s_2,\dots, s_n, k, q)$ будет состояние $(a', s_i, s'_2,\dots,s'_n,k + l, q')$ со следующими свойствами. Обозначим через $(d,p)$ ход соответствующий состоянию $q$, тогда

\begin{itemize}
	\item	$a' = ad$;
	\item 	$s'_i = s_id$ если $i$-ый бит $p$ равен 1 и $s_i = a$, иначе $s'_i = s_i$;
	\item 	$ w\in {\{0,1\}}^{n+2}$, где:
	\begin{itemize}
		\item $i$-ый бит $w$ равен 1 если $s_i = a$, иначе 0;
		\item $n + 1$-ый бит равен 1, если $a \in M$, иначе 0;
		\item $n + 2$-ой бит получается из $\xi_k$
	\end{itemize}
	\item $q'$ получается из состояния $q$ в автомате $R$ по вектору $w$.
\end{itemize}
\end{definition}
В целом, система лабиринт-робот $L$ работает так. Начинаем с начального состояния, и поочерёдно изменяем состояние согласно ходам робота. То есть, можно сказать, что данная система генерирует последовательность состояний $l_k$ согласно вышеописанным правилам. Расположение робота на $k$-ом ходу обозначим через $a_k$. Важно отметить, что вероятности перемещений не зависят от номера хода.

\begin{lemma}
Для задач обхода G равносильно рассматривать существование конечного автомата с генератором случайных чисел и существование недетерминированного автомата с рациональными вероятностями перехода (мы действуем в предположении, что существует $d_id_j = e$).
\end{lemma}

\begin{proof} Если существует автомат с генератором случайных чисел, то существует недетерминированный автомат с рациональными вероятностями перехода, так как, по сути, автомат с генератором случайных чисел --- частный случай недетерминированного автомата с рациональными вероятностями перехода.

Теперь покажем, что утверждение верно и в обратную сторону. Построим автомат с генератором случайных чисел. Там будут состояния, аналогичные состояниям недетерминированного автомата. Если переходы из вершины имели вероятности $(\frac{p_1}{q_1},\frac{p_2}{q_2},\dots,\frac{p_k}{q_k})$, тогда приведём их к общему знаменателю $$(\frac{p'_1}{q},\frac{p'_2}{q},\dots,\frac{p'_k}{q}).$$
Построим переходы с промежуточным состояниями таким образом, чтобы были аналогичные переходы с вероятностью $$(\frac{p'_1}{2^{2q}},\frac{p'_2}{2^{2q}},\dots,\frac{p'_k}{2^{2q}})$$ и возврат в исходную вершину с тем же состоянием с вероятностью $$\frac{2^{2q}-\sum p'_i}{2^{2q}}.$$ Для этого возьмём сбалансированное бинарное дерево из исходной вершины с $2^{2q}$ листами. Ход между ними имеет вид $(d_i,p)$ из вершины нечётного уровня и $(d_j,p)$ из вершины чётного уровня, где $p={\{0\}}^{n}$, а переходы зависят только от случайного бита. Из ${p'}_i$ листовых вершин переход в вершину аналогичной той, в которую был переход с вероятностью $\frac{p_i}{q_i}$ с таким же ходом. Из оставшихся листов однозначные переходы в ещё одну добавленную вершину с ходом $(d_i,p)$, а из нее переход в исходную вершину с ходом $(d_j,p)$. Так как$\frac{2^{2q}-\sum p'_i}{2^{2q}}<1$, то с вероятностью 1 робот в какой-то момент перейдёт в одно из состояний соответствующих переходам недетерминированного автомата из рассматриваемой вершины.\end{proof}

Общий смысл этих определений в том, что у нас есть робот, который является недетерминированным конечным автоматом $R$, $n$ камней, лабиринт $G$ и множество флагов на нем $M$, Робот итерационно переходит по своим состояниям и в соответствии этим переходам ходит по лабиринту. Кроме того, он может носить с собой камни, если находится с ними в одной клетке. По сути, робот является программой с конечной памятью и с возможностью получать случайные биты, что мы используем далее для более удобной демонстрации возможности обхода некоторых лабиринтов. Равносильность этих утверждений расписывать не будем, но задача написать программу по роботу или построить робота по программе не составляет особого труда.

\end{subsection}

\begin{subsection}{Случайные блуждания}

\begin{subsubsection}{Базовые понятия случайных блужданий}
	Рассмотрим некоторые необходимые понятия случайных блужданий:

	\begin{definition}
	Простое дискретное случайное блуждание в $\zs^k$ — это случайный процесс ${\{Y_n\}}_{n\geq0}$ с дискретным временем, имеющий вид
\end{definition}

	\begin{itemize}
		\item $Y_n=Y_0+\mathlarger{\mathlarger{\sum}}_{i=1}^{n}X_i$, где $Y_0$ — начальное состояние ${\{0\}}^k$;
	\item $P(X_i=e_j)=(X_i=-e_j)=\frac{1}{2k}$, где $e_1,\dots e_k$ -вектора естественного ортогонального базиса
	\item Случайные величины $X_i,\ i = 1, 2,\dots$ совместно независимы.
	\end{itemize}
	
	\begin{statement}
	Для случайного блуждания в $\zs^k$ равносильны следующие свойства:
	\end{statement}

	\begin{itemize}
		\item
	 $P$ ($Y_i={\{0\}}^k$ бесконечное число раз$)=1$, иначе говоря
	
	\begin{equation}
\forall l \ P(x:\exists i_1,i_2,\dots,i_l:\ i_1<i_2<\dots<i_l, \forall j \  Y_{i_j}(x)={\{0\}}^k)=1.
	\end{equation}
	
	Будем писать
$P(x:\exists i_1,i_2,\dots,i_l:\ i_1<i_2<\dots<i_l, \forall j \  Y_{i_j}(x)={\{0\}}^k)=1$), как $P(Y_i={\{0\}}^k$ хотя бы $k$ раз$)$,
	\item	$P(x:\exists i > 0:\ Y_i(x)={\{0\}}^k)=1$. Где $x$ элемент вероятностного пространства, на котором задано случайное блуждание). Назовём это свойство {\it возвратностью} случайного блуждания.
	\item $\forall \vec{x} $ Если $\exists i: P(Y_i=\vec{x})> 0$, то $P(Y_i=\vec{x})=1;	P(\exists i>0: Y_i={\{O\}}^k)$
	\end{itemize}

\begin{proof}
	Покажем несколько следствий между этими свойствами,
	\begin{itemize}
	\item $1\Rightarrow2$
	
	Очевидно, так как получается подставлением $l=1$
	\item $2\Rightarrow1$
	
	Продемонстрируем верность этого факта индукцией по $l$. В качестве базы будет свойство 2.
	
	{\bf Переход.} Пусть верно для $l$, докажем для $l+1$.
	
	\begin{equation}
	Q(l,m_1)=P(\exists i_1,\dots,i_l:\ i_1<i_2<\dots<i_l<m_1, \forall j \  Y_{i_j}(x)={\{0\}}^k) \underset{m\rightarrow \infty}{\longrightarrow}1
	\end{equation}
	
	\begin{equation}
	Q(1,m_2)=P(\exists i:\ i<m_2, \  Y_{i}(x)={\{0\}}^k) \underset{m\rightarrow \infty}{\longrightarrow}1
	\end{equation}
	
	Из независимости случайных величин можно вывести, что
	
	\begin{equation}
	Q(l+1,m_1+m_2)\geq Q(l,m_1)Q(1,m_2)
	\end{equation}
	
	Значит, он тоже стремится к 1.
	\item $3\Rightarrow2$
	Берём $\vec{x}={\{0\}}^k$
	\item $1\Rightarrow3$
	
	\begin{equation}
	P(Y_i = {\{0\}}^k\ \mbox{бесконечное\ число\ раз})=1
	\end{equation}
	
	Тогда, если	$\exists i:\ P(Y_i=\vec{x})>0$, то $P(y:\exists i>0:Y_i(y)=\overrightarrow{-x})=1$
	Аналогичные для $\overrightarrow{-x}$ получаем
	$P(y:\exists i>0:Y_i(y)=\vec{x})=1$.
	\end{itemize}

	\end{proof}

	\begin{remark}
	Из этого следует, что возвратность для простого случайного блуждания эквивалентна гарантированному обходу достижимого пространства~\cite{TV3}.
\end{remark}

\begin{remark}
	$P(x:\exists i > 0:\ Y_i(x)={\{0\}}^k)=1$ эквивалентно расходимости ряда
	$\mathlarger{\mathlarger{\sum}}_{i=1}^{\infty}P(Y_i={\{0\}}^k)$.
\end{remark}

\begin{definition}
	Взаимно однозначное отображение $ \pi= (\pi_1,\pi_2,\dots)$ множества (1, 2,\dots) в себя назовём конечной перестановкой, если $pi_n=n$ для всех $n$, за исключением, быть может, конечного числа.
\end{definition}

\begin{definition}	
	Если $\xi=(\xi_1,\xi_2,\dots)$ --- некоторая последовательность случайных величин, то через $\pi(\xi)$ будем обозначать последовательность $\xi=(\xi_{\pi_1},\xi_{\pi_2},\dots)$. Обозначим $\mathscr{F}_n^{\infty}=\sigma(\xi_n,\xi_{n+1},\dots)$
	--- $\sigma$-алгебру, порожденную случайными величинами $\xi_1,\xi_2,\dots$. И пусть $\mathscr{F}'=\mathlarger{\mathlarger{\bigcup}}_{n=1}^{\infty}\mathscr{F}_n^{\infty}$. Поскольку пересечение $\sigma$-алгебр есть $\sigma$-алгебра, то $\mathscr{F}'$ --- это $\sigma$-алгебра. Она называется «хвостовой» или «остаточной».
\end{definition}

\begin{definition}
	Событие $A=\{\xi \in B\}, B \in \mathscr{B}(\mathbb{R}^{\infty})$, то через $\pi(A)$ обозначим событие
	\begin{equation}
	\{\pi(\xi \in B)\},B \in \mathscr{B}(\mathbb{R}^{\infty}).
	\end{equation}
\end{definition}

\begin{definition}
	Событие $A=\{\xi \in B\}, B \in \mathscr{B}(\mathbb{R}^{\infty})$, называется {\em перестановочным}, если для любой конечной перестановки $\pi$ событие $\pi(A)$ совпадает с $A$.
\end{definition}
	
	\begin{theorem} Закон «0 или 1» Хьюитта и Сэвиджа~\cite{TV1}. \label{theo:hyusev}
		
	Пусть $\xi=(\xi_1,\xi_2,\dots)$ - последовательность независимых одинаково распределенных случайных величин и $A = \{\xi \in B\}$-перестановочное событие. Тогда вероятность $P(A)$ может принимать лишь два значения: нуль или единица.
\end{theorem}

\begin{remark}
	Заметим, что свойство {\it блуждания $Y_i$ c шагом $X_i\ A=(Y_i=\{o\}^k)$ бесконечное число раз} является перестановочным событием случайных величин $X_y$. Значит по Теореме~\ref{theo:hyusev} $P(A)$ принимает значение $0$ или $1$,
\end{remark}

\end{subsubsection}

\begin{subsubsection}{Невозвратность случайного блуждания с тремя некомпланарными векторами на $\zs^k$}
Доказательство невозвратности случайного блуждания с тремя некомпланарными векторами на $\zs^k$ разобьём на три части:

\begin{itemize}
\item	Невозвратность простого случайного блуждания на $\zs^3$;
\item Смесь двух случайных блужданий, одно из которых эквивалентно простому случайному блужданию на $\zs^3$ невозвратно;
\item Невозвратность случайного блуждания с тремя некомпланарными векторами на $\zs^k$.
\end{itemize}
	
	\begin{paragraph}{Невозвратность простого случайного блуждания на $\zs^3$.}

Докажем, что $\mathlarger{\mathlarger{\sum}}_{i=1}^{\infty}P(Y_i={\{0\}}^3)$ конечна, так как из этого следует, что

\begin{equation}P(\exists i > 0:\ Y_i(x)={\{0\}}^3)=\varepsilon<1.
\end{equation}

$P(Y_{2n+1}={\{0\}}^k)=0$, потому что мы не можем вернуться в исходную клетку за нечётное число шагов.

\begin{equation}
P(Y_{2n}={\{0\}}^3)=
\mathlarger{\mathlarger{\sum}}_{\substack{0\leq i_1, i_2, i_3\leq n\\ i_1+i_2+i_3=n}}\Big(\frac{(2n)!}{{(i_1!i_2!i_3!)}^2}{\big(\frac{1}{6}\big)}^{2n}\Big)=
\end{equation}

\begin{equation}=2^{-2n}\cdot C^{n}_{2n} \cdot
\mathlarger{\mathlarger{\sum}}_{\substack{0\leq i_1, i_2, i_3\leq n\\ i_1+i_2+i_3=n}}{\Big(\frac{n!}{i_1!i_2!i_3!}{\big(\frac{1}{3}\big)}^{n}\Big)}^2\leq
\end{equation}

\begin{equation}\leq 2^{-2n}\cdot C^{n}_{2n} \cdot C_{n} \cdot
\mathlarger{\mathlarger{\sum}}_{\substack{0\leq i_1, i_2, i_3\leq n\\ i_1+i_2+i_3=n}}\frac{n!}{i_1!i_2!i_3!}{\big(\frac{1}{3}\big)}^{2n}
\end{equation}
где $C_n=\underset{\substack{0\leq i_1, i_2, i_3\leq n\\ i_1+i_2+i_3=n}}{\max}\frac{n!}{i_1!i_2!i_3!}$, тогда

\begin{equation}
P(Y_{2n}={\{0\}}^3)\leq 2^{-2n}\cdot C^{n}_{2n} \cdot C_{n} \cdot 3^{-n}\cdot
\mathlarger{\mathlarger{\sum}}_{\substack{0\leq i_1, i_2, i_3\leq n\\ i_1+i_2+i_3=n}}\frac{n!}{i_1!i_2!i_3!}{\big(\frac{1}{3}\big)}^{n}=2^{-2n}\cdot C^{n}_{2n} \cdot C_{n} \cdot 3^{-n}
\end{equation}

Покажем, что $$C_n\sim \frac{n!}{\lceil {\frac{n}{3}+1}\rceil!^3} \sim \frac{n!}{\lfloor \frac{n}{3}!\rfloor^3}.$$ Ясно, что $$C_n \geq \frac{n!}{\lfloor \frac{n}{3}!\rfloor^3}$$ (возьмём $i_j\geq \lfloor \frac{n}{3}\rfloor$). Кроме того, при $i_1>\lfloor \frac{n}{3}\rfloor>i_2$  верно, что $i_1!\cdot i_2!>(i-1)! \cdot (j+1)!$. Увеличивая подобным образом $\frac{n!}{i_1!i_2!i_3!}$ мы остановимся, когда все $i_j$ будут равны $\lfloor \frac{n}{3}\rfloor$ или $\lceil \frac{n}{3}\rceil.$

Применив к
\begin{equation}
2^{-2n}\cdot C^{n}_{2n} \cdot C_{n} \cdot 3^{-n}
\end{equation}
формулу Стирлинга, получим $\frac{3\sqrt{3}}{2 \cdot \pi^{1.5} \cdot n^{1.5}}$.

Сумма $\sum \frac{3\sqrt{3}}{2 \cdot \pi^{1.5} \cdot n^{1.5}}$ сходится, значит $P(Y_{2n}={\{0\}}^3)$ ограничена функцией, сумма ряда которой сходится. Но, тогда она сама сходится, и блуждание является невозвратным.

\begin{statement}
	$$P(Y_n(\vec{x})\leq 6^3 \cdot (P(Y_n={\{0\}}^3)+P(Y_{n+1}={\{0\}}^3)+P(Y_{n+2}={\{0\}}^3)+P(Y_{n+3}={\{0\}}^3))$$
\end{statement}

\begin{proof}
	Обозначим координаты $\vec{x}=(x_1,x_2,x_3)$, $x'_i=x_i mod 2$, $x'=\mathlarger{\mathlarger{\sum}}_{i=1}^{3} x'_i$. Без ограничения общности можно считать, что $x_i \geq 0$. Тогда

\begin{equation}
	P(Y_{2n+x_1+x_2+x_3}=\vec{x})=
\end{equation}

\begin{equation}
=\mathlarger{\mathlarger{\sum}}_{\substack{0\leq i_1, i_2, i_3\leq n\\ i_1+i_2+i_3=n}}\Big(\frac{(2n+{\mathlarger{\sum}}_{j=1}^{j \leq 3}x_j)!}{{\mathlarger{\prod}}_{j=1}^{j \leq 3}(i_j!\cdot(i_j+x_j)!)}{\big(\frac{1}{6}\big)}^{2n+{\mathlarger{\sum}}_{j=1}^{j \leq 3}x_j}\Big)\leq
\end{equation}

\begin{equation}
\leq 6^3 \cdot \mathlarger{\mathlarger{\sum}}_{\substack{0\leq i_1, i_2, i_3\leq n\\ i_1+i_2+i_3=n}}\Big(\frac{(2n+x'+{\mathlarger{\sum}}_{j=1}^{j \leq 3}x_j)!}{{\mathlarger{\prod}}_{j=1}^{j \leq 3}(i_j!\cdot(i_j+x_j+x'_j)!)}{\big(\frac{1}{6}\big)}^{2n+x'+{\mathlarger{\sum}}_{j=1}^{j \leq 3}x_j}\Big)\leq
\end{equation}
	
	\begin{equation}
	\leq 6^3 \cdot \mathlarger{\mathlarger{\sum}}_{\substack{0\leq i_1, i_2, i_3\leq n\\ i_1+i_2+i_3=n}}\Big(\frac{(2n+x'+{\mathlarger{\sum}}_{j=1}^{j \leq 3}x_j)!}{{\big({\mathlarger{\prod}}_{j=1}^{j \leq 3}(i_j+(x_j+x'_j)/2)!\big)}^2}{\big(\frac{1}{6}\big)}^{2n+x'+{\mathlarger{\sum}}_{j=1}^{j \leq 3}x_j}\Big)\leq
	\end{equation}
	
	\begin{equation}
	\leq 6^3 \cdot
	\mathlarger{\mathlarger{\sum}}_{\substack{0\leq i_1, i_2, i_3\leq n'\\ i_1+i_2+i_3=n'}}\Big(\frac{(2n')!}{{(i_1!i_2!i_3!)}^2}{\big(\frac{1}{6}\big)}^{2n'}\Big)=6^3\cdot P(Y_{n'}={\{0\}}^3)
	\end{equation}
	
	Где $n' = n + (\mathlarger{\mathlarger{\sum}}_{j=1}^{j \leq 3}(x_j+x'j))/2$.
	Тогда $2n' - 2n - (x_1+ x_2  +  x_3 ) =$ 0, 1,2 или 3.
\end{proof}

\begin{remark}
	Из этого следует, что $P(Y_n =\overrightarrow{-x})<\frac{c}{n\sqrt{n}}$, где $c$ --- некоторая константа, одинаковая для всех $\vec{x}$.
\end{remark}
	
\end{paragraph}

\begin{paragraph}{Смесь двух блужданий, одно из которых эквивалентно случайному блужданию на $\zs^3$ невозвратно.} \
	
Смесь блужданий $X$ и $Y$, с вероятностями $p, q$ $(p+q=1; p, q>0)$, где $X$~-- эквивалентен простому случайному блужданию, обозначим через $Z$, Тогда

	\begin{equation}
	P(Z_n=\vec{z})=\sum_{\vec{x}+\vec{y}=\vec{z}}\sum_{i=0}^{n}P(X_i=\vec{x})\cdot P(Y_{n-i}=\vec{y})\cdot p^i \cdot q^{n-i} \cdot C_n^i .
	\end{equation}
	
Мы хотим доказать, что $P(Z_n = \vec{z}) < \frac{b}{n\sqrt{n}}$ для некоторого $b$. Обозначим за $\varepsilon = \min(p, q)^2$ и докажем, что такие константы b есть для

\begin{equation}
\sum_{\vec{x}+\vec{y}=\vec{z}}\sum_{i=0}^{i<\varepsilon n}P(X_i=\vec{x})\cdot P(Y_{n-i}=\vec{y})\cdot p^i \cdot q^{n-i} \cdot C_n^i
\end{equation}
и

\begin{equation}
\sum_{\vec{x}+\vec{y}=\vec{z}}\sum_{i\geq \varepsilon n}^{n}P(X_i=\vec{x})\cdot P(Y_{n-i}=\vec{y})\cdot p^i \cdot q^{n-i} \cdot C_n^i.
\end{equation}

	Начнём cо второго случая

	\begin{equation}
	\sum_{\vec{x}+\vec{y}=\vec{z}}\sum_{i\geq \varepsilon n}^{n}P(X_i=\vec{x})\cdot P(Y_{n-i}=\vec{y})\cdot p^i \cdot q^{n-i} \cdot C_n^i\leq
	\end{equation}

	\begin{equation}
	\leq \max_{i\geq\varepsilon n,\vec{x}}(P(X_i=\vec{x}))\sum_{\vec{x}+\vec{y}=\vec{z}}\sum_{i\geq \varepsilon n}^{n}\cdot P(Y_{n-i}=\vec{y})\cdot p^i \cdot q^{n-i} \cdot C_n^i=
	\end{equation}

	\begin{equation}
	= \max_{i\geq\varepsilon n,\vec{x}}(P(X_i=\vec{x}))\sum_{i\geq \varepsilon n}^{n}\sum_{\vec{x}+\vec{y}=\vec{z}} P(Y_{n-i}=\vec{y})\cdot p^i \cdot q^{n-i} \cdot C_n^i=
	\end{equation}

	\begin{equation}
	= \max_{i\geq\varepsilon n,\vec{x}}(P(X_i=\vec{x}))\sum_{i\geq \varepsilon n}^{n} p^i \cdot q^{n-i} \cdot C_n^i \cdot \sum_{\vec{x}+\vec{y}=\vec{z}} P(Y_{n-i}=\vec{y})\leq
	\end{equation}

	\begin{equation}
	\leq \max_{i\geq\varepsilon n,\vec{x}}(P(X_i=\vec{x}))\sum_{i\geq \varepsilon n}^{n} p^i \cdot q^{n-i} \cdot C_n^i\leq
	\end{equation}

	\begin{equation}
	\leq \max_{i\geq\varepsilon n,\vec{x}}(P(X_i=\vec{x}))\cdot {p+q}^n=\max_{i\geq\varepsilon n,\vec{x}}(P(X_i=\vec{x}))\leq \frac{c}{\varepsilon n \sqrt{\varepsilon n}}=\frac{c_2}{n\sqrt{n}}.
	\end{equation}

	Теперь разберём первый случай

	\begin{equation}
	\sum_{\vec{x}+\vec{y}=\vec{z}}\sum_{i=0}^{i<\varepsilon n}P(X_i=\vec{x})\cdot P(Y_{n-i}=\vec{y})\cdot p^i \cdot q^{n-i} \cdot C_n^i.
	\end{equation}

	Обозначим через $p'=\min(p,q),q'=1-p'$. Тогда, так как $\varepsilon \leq \frac{1}{4}$, то

	\begin{equation}
	\sum_{\vec{x}+\vec{y}=\vec{z}}\sum_{i=0}^{i<\varepsilon n}P(X_i=\vec{x})\cdot P(Y_{n-i}=\vec{y})\cdot p^i \cdot q^{n-i} \cdot C_n^i\leq
	\end{equation}

	\begin{equation}
	\leq\sum_{\vec{x}+\vec{y}=\vec{z}}\sum_{i=0}^{i<\varepsilon n}P(X_i=\vec{x})\cdot P(Y_{n-i}=\vec{y})\cdot p'^i \cdot q'^{n-i} \cdot C_n^i
	\end{equation}

	Зададим $f(i)=p'^i\cdot q'^{n-i}\cdot c^n_i$. Тогда $f(i)$ возрастает на $(0,\varepsilon n)$.

	\begin{equation}
	\frac{f(i+1)}{f(i)}=\frac{p'}{q'}\cdot\frac{n-i}{i+1}\geq\frac{p'}{q'}\cdot\frac{n-\varepsilon n}{\varepsilon n+1}\geq \frac{p' \cdot (1-p'^2)}{p'^2 \cdot (1-p')}=1+\frac{1}{p'}
	\end{equation}

	Значит, с какого-то $n$ можно сказать, что $f(i)$ возрастает на $(0,\varepsilon n)$.	Тогда

	\begin{equation}
	\sum_{\vec{x}+\vec{y}=\vec{z}}\sum_{i=0}^{i<\varepsilon n}P(X_i=\vec{x})\cdot P(Y_{n-i}=\vec{y})\cdot p'^i \cdot q'^{n-i} \cdot C_n^i\leq
	\end{equation}

	\begin{equation}
	\leq f(\lfloor\varepsilon n\rfloor)\cdot\sum_{i=0}^{i<\varepsilon n} \sum_{\vec{x}+\vec{y}=\vec{z}}P(X_i=\vec{x})\cdot P(Y_{n-i}=\vec{y})\leq
	\end{equation}

	\begin{equation}
	\leq f(\lfloor\varepsilon n\rfloor)\cdot\varepsilon n
	\end{equation}

	Так как $\sum_{\vec{x}+\vec{y}=\vec{z}}P(X_i=\vec{x})\cdot P(Y_{n-i}=\vec{y})\leq\sum_{\vec{x}}P(X_i=\vec{x})\cdot \sum_{\vec{y}}P(Y_i=\vec{y})=1$.

	\begin{equation}
	f(\lfloor\varepsilon n\rfloor)\cdot\varepsilon n=p'^2 \cdot n\cdot C^{n\cdot p'^2}_n \cdot {p'}^{np'^2}\cdot(1-p')^{n(1-p'^2)}
	\end{equation}

	Применив сюда формулу Стирлинга, получим

	\begin{equation}
	f(\lfloor\varepsilon n\rfloor)\cdot\varepsilon n\sim const \cdot \frac{n}{\sqrt{n}}\cdot\Big(\frac{p'^{p'^2}\cdot (1-p')^{1-p'^2}}{(p'^2)^{p'^2}\cdot (1-p'^2)^{1-p'^2}}\Big)^n
	\end{equation}

	Если мы докажем, что $d=\frac{p'^{p'^2}\cdot(1-p')^{1-p'^2}}{(p'^2)^{p'^2}\cdot (1-p'^2)^{1-p'^2}}<1$, тогда так как верно, что $\forall c$$\forall d'<0$$ \exists n_0 \forall n>n_0$

	\begin{equation}
	const \sqrt{n}\cdot d'^n\leq c \cdot \frac{1}{n \sqrt{n}}
	\end{equation}

	Мы получим, что $\exists n_0,c': \forall n>n_0$

	\begin{equation}
	\sum_{\vec{x}+\vec{y}=\vec{z}}\sum_{i=0}^{i<\varepsilon n}P(X_i=\vec{x})\cdot P(Y_{n-i}=\vec{y})\cdot p^i \cdot q^{n-i} \cdot C_n^i\leq \frac{c'}{n\sqrt{n}}
	\end{equation}

	Так как эта сумма сама по себе не превосходит единицу, то $\exists n_0,c_1: \forall n>0$

	\begin{equation}
	\sum_{\vec{x}+\vec{y}=\vec{z}}\sum_{i=0}^{i<\varepsilon n}P(X_i=\vec{x})\cdot P(Y_{n-i}=\vec{y})\cdot p^i \cdot q^{n-i} \cdot C_n^i\leq \frac{c_1+n_0\sqrt{n_0}}{n\sqrt{n}}
	\end{equation}

	Осталось проверить утверждение про d:
	\begin{equation}
	d<1\Leftrightarrow \frac{p'^{p'^2}\cdot(1-p')^{1-p'^2}}{(p'^2)^{p'^2}\cdot (1-p'^2)^{1-p'^2}}<1\Leftrightarrow
	\end{equation}

	\begin{equation}
	\Leftrightarrow p'^{p'^2}\cdot(1-p')^{1-p'^2}<(p'^2)^{p'^2}\cdot (1-p'^2)^{1-p'^2}\Leftrightarrow  p'^{(-p'^2)}<(1+p')^{1-p'^2}\Leftrightarrow
	\end{equation}

	\begin{equation}
	\Leftrightarrow
	\big(1+\frac{1}{p'}\big)^{(p'^2)}<1+p'\Leftrightarrow
	p'^2\cdot \ln\big(1+\frac{1}{p'}\big)<\ln(1+p')\Leftrightarrow
	\end{equation}

	\begin{equation}
	\Leftrightarrow
	\frac{\ln(1+p')}{p'}<p'\cdot{\ln\big(1+\frac{1}{p'}\big)}
	\end{equation}

	Заметим, что последнее это $g(1+p')>g\big(1+\frac{1}{p'}\big)$, где $g(x)=\frac{\ln(1+x)}{x}$. Нам достаточно показать, что $g(x)$  убывающая функция при $x>0$.

	\begin{equation} g'(x)<0
	\Leftrightarrow
	\frac{\frac{x}{1+x}-\ln(x+1)}{x^2}<0
	\Leftrightarrow
	\frac{x}{1+x}-\ln(x+1)<0
	\Leftrightarrow
	\end{equation}

	\begin{equation} g'(x)<0
	\Leftrightarrow
	0<\ln(x+1)-1+\frac{1}{1+x}=h(x)
	\end{equation}

	$h(0)=0$ и $h'(x)=\frac{1}{1+x}-\frac{1}{(1+x)^2}=\frac{x}{(1+x)^2}$, что больше нуля при $x>0$, тогда $h(x)>0$.
	
	Из разбора этих двух случаев становится ясно, что $P(Z_n=\vec{z})\leq \frac{c_1+c_2+n_0\vec{n_0}}{n\sqrt{n}}$.
\end{paragraph}

\begin{paragraph}{Невозвратность случайного блуждания с тремя некомпланарными векторами.}

Докажем это утверждение на $\zs^k$.

\begin{remark}
	Если бы это было не так, то из любой достижимой клетки можно было бы вернуться в стартовую с положительной вероятностью.
\end{remark}

\begin{remark}
Если блуждание $Y_i$ построенное на шагах $X_i$ возвратное, то $Y^{0}_i$ которое строится на нем с шагами равным $n$ обычных ходов $X'_i=\sum_{j=n\cdot (i-1)+1}^{n\cdot i}X_i$, то есть $Y^{0}_i=Y_{ni}$ тоже возвратное.
\end{remark}

\begin{proof}
	Если бы это было не так, то $P(Y^{0}_i ={\{0\}}^k )$ бесконечное число раз$)=0$ (по закону 0 и 1), кроме того, рассмотрим блуждания $Y^{l}_i=Y_{ni+l}$ , при $l<n$. Какое-то из них посещает ${\{0\}}^k$ бесконечное число раз, но тогда рассмотрим его после первого посещения ${\{0\}}^k$. Оно будет эквивалентно $P(Y^{0}_i ={\{0\}}^k $ бесконечное число раз$)$. Противоречие.
\end{proof}

 Пусть случайное блуждания $Y_i$ с тремя некомпланарными векторами на $\zs^k$ возвратно (обозначим эти вектора $\vec{a_1}, \vec{a_2}, \vec{a_3}$).
  Тогда $\exists n_1>1: P(Y_{n_1-1}=\vec{-a_1})>0$, кроме того $P(Y_{n_1-1}=(n_1-1)\cdot \vec{a_1})>0$ ($n_1-1$ шаг по $-\vec{a}$).
  Из этого следует, что для $m_1=n_1\cdot(n_1-1)^2,\ l_1=(n_1-1)\cdot n_1$ верно, что $P(Y_{m_1}=-l_1\cdot\vec{a_1})>0$ (делаем $n_1\cdot(n_1-1)$ раз шаг за $n_1-1$ ходов дающий $\overrightarrow{-a}$ ) и  $P(Y_{m_1}=l_1\cdot\vec{a_1})>0$
  ( делаем $(n_1-2)\cdot(n_1-1)$ раз шаг за $n_1-1$ ходов дающий $\overrightarrow{-a_1}$ и $2(n_1-1)^2 $ шагов по вектору $\vec{a_1}$). Проведя аналогичные действия получим
  $Y'_i=Y_{m_1\cdot m_2\cdot m_3}$ с положительными вероятностями переходов по векторам
  $$l_1\cdot m_2 \cdot m_3 \cdot \vec{a_1}, -l_1\cdot m_2 \cdot m_3 \cdot \vec{a_1},
  m_1\cdot l_2 \cdot m_3 \cdot \vec{a_2}, -m_1\cdot l_2 \cdot m_3 \cdot \vec{a_2},
  m_1\cdot m_2 \cdot l_3 \cdot \vec{a_3}, -m_1\cdot m_2 \cdot l_3 \cdot \vec{a_3}.$$
  Но тогда $Y'_i$ --- смесь блуждания эквивалентного простому в $\zs^3$ c коэффициентом минимальной вероятности одного из этих 6 шагов( а она положительная) умноженным на 6 и еще некоторого блуждания. Значит оно невозвратно. Противоречие.

  \begin{remark}
  	Для случайного блуждания с тремя некомпланарными векторами (и ненулевой вероятностью вернуться в стартовую клетку из любой достигаемой) $Y_i$ $\forall \varepsilon>0 \exists$ конечное число $\vec{x}: P(\exists i: Y_i={\vec{x}})>\varepsilon$.
  \end{remark}

Это легко доказать, показав, что для $$Y_i\ \exists c\ \forall \vec{x}\ \forall n \  P(Y_n=\vec{x})<\frac{c}{n\sqrt{n}}.$$ Из этого следует, что $$\exists n_0\ \forall \vec{x} : P(\exists n>n_0 :Y_n=\vec{x})<\frac{\varepsilon}{2},$$ но количество $\vec{x}:P(\exists n\leq n_0 :Y_n=\vec{x})>=\frac{\varepsilon}{2}$ не больше $ \frac{2n_0}{\varepsilon}$. Значит, количество клеток, которые будут посещены с вероятностью хотя бы $\varepsilon$ не больше $\frac{2n_0}{\varepsilon}$.

\end{paragraph}

\end{subsubsection}

\end{subsection}
	
\section{Обход пространства размерностью меньше 10.}

\begin{subsection}{Робот с генератором случайных битов.}

	Разберём для начала случай, когда есть только робот, и покажем, что он может обойти $\zs^2$ и не может $\zs^3$ . В этом случае $M=\{\mbox{\o}\}, n=0$. Множество $D$ в пространстве $\zs^k$ имеет вид $(\vec{e_1},-\vec{e_1},\vec{e_2},-\vec{e_2},\dots,\vec{e_k},-\vec{e_k})$. Робот считается обходящим пространство, если для любой ее клетки он посещает её с вероятностью один.

\begin{subsubsection}{Обход роботом $\zs^2$.}
	Написать программу для такого робота довольно легко. Он просто должен эмулировать случайное блуждание. А так как оно обходит плоскость, то и робот ее обойдёт~\cite{TV2}.
\end{subsubsection}

\begin{subsubsection}{Невозможность обойти $\zs^3$.}
	Докажем это для любого робота. Рассмотрим граф, соответствующий недетерминированному конечному автомату. Обозначим его размер через $m$. Будем считать, что вероятность каждого ребра больше нуля и все вершины достижимы из начальной (кроме, возможно, её самой). Иначе их можно просто выкинуть и это никак не повлияет на поведение робота. Возьмём компоненты сильной связности графа. Назовём компоненту листовой, если из неё нет рёбер в другие компоненты (такая обязательно есть). Клетками будем называть элементы лабиринта (в данном случае элементы пространства  $\zs^k$), вершинами состояния робота. Расстояние между клетками будет считать по метрике расстояния городских кварталов. Плоскостями будем называть множество клеток вида $\vec{x_0}+a_1\vec{x_1}+a_2\vec{x_2}$	при фиксированных $\vec{x_0},\ \vec{x_1},\ \vec{x_2}$ и целых $a_1,\ a_2$. Предположим, что есть робот, обходящий пространство.

\begin{statement}
	Робот с недетерминированным графом, соответствующим листовой компоненте сильной связанности робота, обходящего пространство, тоже должен обходить пространство.
\end{statement}

\begin{proof}
	Обозначим какую-то из его вершин за новое начальное состояние $q'_0$, а новый автомат, построенный на этой листовой компоненте за $R'_0(q'_0)$, размер компоненты обозначим через $m'$. Тогда новый робот должен обходить все пространство, кроме не более чем $m$ клеток. Это связано с тем, что от $q_0$ до $q'_0$ есть путь не более, чем за $m$ шагов (обозначим за $\vec{x_1}(q'_0)$ клетку в пространстве в которой мы оказались дойдя до $(q'_0)$). Тогда продолжение этого пути с вероятностью 1 должно побывать во всех клетках пространства, кроме тех, в которой он уже был, а их не более $m$. Так как они последовательны, то эти клетки помещаются в сферу с центром в ${\{0\}}^k$ (где $k$ --- размерность пространства, в котором мы работаем) и радиусом $m$, и в сферу с центром $\vec{x_1}(q'_0)$ и радиусом $m$. Это верно при выборе любого $q'_0$ из нашей листовой компоненты. С вероятностью один мы когда-нибудь попадём в клетку $\vec{x_2}$ на расстоянии $2m+1$ от ${\{0\}}^k$. Обозначим состояние, в котором мы оказались в этой клетке за $q'_1$ (Это состояние принадлежит нашей листовой компоненте). С этого момента робот ведёт себя как робот $R'(q'_1)$, находящийся в клетке $\vec{x_2}$ в изначальном состоянии, а он обходит все пространство, кроме сферы с центром в $\vec{x_2}$ и радиусом $m$. Значит $R'(q'_0)$ обходит с одной стороны все клетки кроме сферы с центром в ${\{0\}}^k$ и радиусом $m$, с другой стороны все клетки кроме сферы с центром в $\vec{x_2}$ и радиусом $m$. Эти сферы не пересекаются, поэтому $R'_0(q'_0)$ обходит все пространство. Кроме того, перемещения этого робота являются возвратными.
\end{proof}
	
	Теперь разберёмся, почему недетерминированный автомат, ориентированный граф которого является компонентой сильной связности, не может обходить пространство (начальная вершина --- $q_0$, размер графа $m$). Обозначим за $p_{i,j}$ вероятность попасть в состояние $q_j$ из состояния $q_i$ за количество шагов, не превосходящее $m$ (при $i=j$ считаем, что нужен хотя бы один шаг). Так как это компонента сильной связности c m вершинами $p_{i,j}>0$. Обозначим за p = $\mathlarger{\min_{i,j}}\ p_{i,j}$. Тогда вероятность не попасть в течение $l\cdot m$ шагов в $j$-ую вершину меньше ${(1 - p)}^l$. А оно стремится к нулю при $l\rightarrow\infty$. Значит, с вероятностью 1 мы побываем в состоянии $q_j$ из чего следует, что мы побываем там бесконечное число раз.
	
	Мы уже получили, что перемещения такого робота возвратны, то есть он побывает в клетке ${\{0\}}^k$ бесконечное число раз. Рассмотрим вероятности попасть в клетку в состояние $q_j$, если прошлый раз мы были в ней в состояние $q_i$. Обозначим эту вероятность через $p_{i,j}$ ($p_{i,j}$ не зависит от клетки в которой мы начинаем). $\forall i \sum_{j=1}^{m} p_{i,j}= 1$. Рассмотрим эти переходы по состояниям, как некоторый другой ориентированный граф, по которому мы гуляем бесконечно долго. В нем можно выбрать несколько состояний $q_{i,1},q_{i,2},\dots,q_{i,l}$ таких, что мы гарантировано побываем хотя бы в одном из них, они все достижимы из $q_0$ и $q_{i,j_1}$ не достижимо из  $q_{i,j_2}$ (будем брать по одной вершине из листовых компонент сильной связности нового ориентированного графа). Мы побываем в одном из этих состояний бесконечное число раз. Обозначим его через $q'$.

	\begin{remark}
		Для любой клетки, где робот был в состояние $q'$, он побывал там бесконечное число раз.
	\end{remark}

	Теперь рассмотрим случайное блуждание в пространстве с векторами переходов и их вероятностями такими, что пара (вектор, вероятность) соответствуют паре (вектор, вероятность) перемещения робота между двумя состояниями $q'$. Как сказано выше, это случайное блуждание возвратно, значит оно не может содержать некомпланарных векторов. Тогда после первого попадания в $q'$ множество клеток, где мы можем быть в состояние $q'$ является элементом какой-то плоскости. Но мы гарантировано побываем в любой клетке, а в сфере с центром в ней и радиусом $m$ есть вероятность побывать в состояние $q'$. Тогда если мы возьмём клетку на расстояние больше $m+1$ от этой плоскости, то мы не сможем в неё попасть. Противоречие.
\end{subsubsection}

\end{subsection}

\begin{subsection}{Робот с генератором случайных чисел и камнем}
	Теперь рассмотрим случай робота с камнем. Изначальное положение его и камня в клетке с нулевыми координатами.

\begin{subsubsection}{Обход $\zs^4$}
	Довольно легко описать программу, в соответствии с которой будет перемещаться робот. Он отходит от камня, случайно блуждает вдоль координат $x_1$, $x_2$, пока не вернётся к камню, а потом делает шаг в случайном направлении вдоль $x_3$, $x_4$ вместе с камнем. Повторяет.
	Так как случайное блуждание на плоскости возвратно, робот все время возвращается к камню. Тогда, ходя вместе с камнем, робот обходит всю плоскость $x_3$, $x_4$. так как перемещения с камнем в случайном направлении тоже является случайным блужданием. Оно возвратно, поэтому робот побывает с камнем во всех клетках плоскости $x_3$, $x_4$ бесконечное число раз. Тогда для любой из этих клеток мы бесконечное число раз блуждали от камня вдоль $x_1$, $x_2$. Рассмотрев только эти ходы получим простое случайное блуждания из клетки плоскости  $x_3$, $x_4$ вдоль $x_1$, $x_2$, а таким образом можно получить все клетки пространства.
\end{subsubsection}

\begin{subsubsection}{Невозможность обойти $\zs^5$}
	Разберёмся, как ходит робот. Разделим его перемещения на два типа:
	\begin{itemize}
		\item Перемещение с камнем,
		\item Перемещение без камня.
	\end{itemize}

	  \begin{paragraph}{Перемещение без камня.}

	  	\begin{statement}
	  		Если робот, обходящий пространство, уходит от камня, то он должен к нему вернуться с вероятностью один
	  		\end{statement}

  		\begin{proof}
  			Робот, который с какого-то момента имеет вероятность не вернуться, с этого момента соответствует какому-то роботу без камня. А про них уже доказано, что они обходят в лучшем случае плоскость и клетки, удалённые от неё не более, чем на $m$(где $m$-количество состояний робота). А это --- явно не всё пространство.
  		\end{proof}

  		Покажем, что если робот возвращается к камню с вероятностью один, то множество клеток, которые он может посетить, отойдя от камня, является конечным объединением плоскостей. Пусть мы отходим от камня в состоянии $q_1$. Тогда возьмём недетерминированный конечный автомат $R'$ с состояниями, аналогичными $R$ и новым начальным состоянием $q'_1$; вероятностям перехода между состояниями, соответствующими $R$ совпадают с аналогичными в $R$ при условии отсутствия камней в клетке; переход из $q'_1$ соответствуют переходам из $q_1$ при условие наличия камня в клетке. Обозначим количество состояний в $R'$ через $m$. Новый робот имеет те же вероятности положения в пространстве, начиная с отхода от камня в состояние $q'_1$ до возвращения к нему.
	  	
	  	Рассмотрим его ориентированный граф компонент сильной связности.
	  	 Если какой-нибудь лист не достижим при блуждании в пространстве до первого возвращения в изначальную клетку, то мы можем его выкинуть из графа, так как мы все равно в него не попадаем.
	  	  Выкинув таким образом все недостижимые компоненты, возьмём какое-нибудь состояние $q\neq q'_1$.

	  	  \begin{statement}
Множество клеток, в которых мы можем оказаться в состоянии $q$, будет подмножеством объединения конечного числа плоскостей с одинаковыми образующими векторами.
	  	  \end{statement}

	  	   Этот факт очевиден для случая, когда таких клеток конечно, так как их можно покрыть конечным количеством плоскостей (по плоскости на каждую клетку), так что будем рассматривать только $q$, для которых существует сколь угодно далёкая клетка с вероятностью её посетить в состоянии $q$.
	  	
	  	   Докажем это сначала для случая, где $q$ находится в листовой компоненте.
	  	    Возьмём случайное блуждание, построенное на переходах в пространстве робота между двумя состояниями $q$ (между ними могут быть только состояния отличные от $q$).
	  	
	  	    Если в нём нет трёх некомпланарных векторов, то все клетки, куда робот может попасть в этом состоянии лежат в какой-то одной плоскости (даже без условия остановки на камне). Кроме того, от остальных состояний он может дойти до этого за не более, чем $m$ ходов. Тогда множество клеток, где робот может оказаться в состояние $q$ удалено от стартовой клетки не больше, чем на $m$. Значит, состояние $q$ достижимо роботом только в клетках являющимися параллельными плоскостями, удалёнными от стартовой клетки не больше, чем на $m$, а таких конечно.
	  	
	  	    Для блуждания с тремя некомпланарными векторами $\forall \varepsilon>0$ существует конечное число клеток, достижимых с вероятностью хотя бы $\varepsilon$.
	  	    Тогда $\exists r(\varepsilon):$ для клеток удалённых от места нынешнего нахождения хотя бы на $r(\varepsilon)$ вероятность в них оказаться меньше $\varepsilon$.
	  	    Обозначим за $p$ минимальную вероятность дойти из состояния того же листа $q'$ до $q$ за не более, чем $m$ шагов.
	  	    Значит, есть клетка($\vec{x_1}$), достижимая блужданием на расстоянии не больше $m$ от изначальной клетки с вероятностью хотя бы $p$, но тогда она достижима из любой клетки в состоянии $q$ с вероятностью хотя бы $p$.
	  	    Кроме того, есть клетка($\vec{x_2}$), в которой робот может оказаться в состоянии $q$ на расстоянии от изначальной клетки($\vec{x_0}$) большем, чем $m+r(p)$.
	  	    Тогда расстояние между ($\vec{x_2}$) и ($\vec{x_1}$) хотя бы $r(p)$, но в этом случае вероятность попасть в ($\vec{x_1}$) из ($\vec{x_2}$) меньше $p$. Противоречие.

	  	        \begin{remark}
	  	        	Если компонента с состоянием $q$ не листовая, то из неё за не более, чем $m$ шагов можно дойти до состояния $q'$. Тогда клетки, где мы можем оказаться в состоянии $q$ находятся на расстоянии не больше $m$ от подмножества клеток объединения конечного числа плоскостей, что тоже является подмножеством клеток объединения конечного числа плоскостей.
	  	        \end{remark}

  	        Возьмём множество плоскостей $L_i$, соответствующих плоскостям состояний $q$. Кроме того, разрешим $L_i$ совпадать, чтобы плоскость учитывалась отдельно для каждого промежуточного состояния $q$ и $q_1$ из которого мы стартовали от камня. Их количество обозначим за $l$.
	  	\end{paragraph}

  	\begin{paragraph}{Перемещения с камнем.}
	Осталось разобраться, как устроены перемещения с камнем. Построим автомат, соответствующий тасканию камня изначального робота. Для этого состояние в котором мы уходим из клетки оставив камень будем рассматривать, как переход из этого состояния в состояния, в которых мы могли бы вернуться к камню, с вероятностями, с которыми это могло бы произойти. Чтобы мы не оставляли камень, добавим по одному состоянию на каждый такой переход, чтобы сделать шаг вверх с камнем и шаг вниз с ним же. Так как этот автомат все время таскает камень, то он ведёт себя, как просто робот без камня.
	
	Обозначим множество клеток, где побывает этот робот за $B$. Рассмотрим какую-нибудь из его листовых компонент связности. Мы можем в неё попасть за какое-то конечное число ходов робота. Обозначим $B'$ множество клеток, где побывает листовой робот. Рассмотрим вероятности побывать в клетках $\zs^5$ в состояние $q$ на плоскости $L_i$.

\begin{statement}
	Покажем, что $\forall \varepsilon>0$ множество клеток $x_j$ для которых $P (\exists i:$листовой робот попал в $L_i(x_j)$ в состояние $q)>\varepsilon$ имеет вид конечного объединения подпространств размерности четыре, где за подпространство размерности четыре берём множество клеток представимых в виде $\vec{x_0}+a_1\vec{x_1}+a_2\vec{x_2}+a_3\vec{x_3}+a_4\vec{x_4}$	при фиксированных линейных независимых $\vec{x_0},\ \vec{x_1},\ \vec{x_2},\ \vec{x_3},\ \vec{x_4}$ и целых $a_1,\ a_2,\ a_3,\ a_4$.
\end{statement}

\begin{proof}
	Давайте посмотрим на множество клеток, в которые он попадёт при блуждание по плоскостям $L_i$ с вероятностью хотя бы $\frac{\varepsilon}{l}$.
	Для этого введём на $\zs^5$ отношение эквивалентности $\vec{x_1}\sim\vec{x_2}$, которое верно при $L_i(\vec{x_1})=L_i(\vec{x_2})$.
	Посмотрим, что произойдёт с блужданием на склейке.
	Обозначим новый лабиринт через $Z'$, а его случайное блуждание через $X'$.
	Если $X'$ возвратное, то достижимые клетки задаются не более чем двумя векторами, но тогда $L_i(B)$ уже подпространство размерности четыре.
	Из невозвратности $X'$ следует, что в $Z'$ конечное число клеток с вероятностью попадания хотя бы $\frac{\varepsilon}{l}$.
	Развернув отображение обратно получим, что множество клеток $\zs^5$, в которые робот попадёт при блуждание от камня соответствущее плоскости $L_i$ с вероятностью хотя бы $\frac{\varepsilon}{l}$ будет конечным объединением плоскостей. Так как вероятность робота попасть способом, соответствующим $L_i$, в прообраз не превосходит вероятности из образа.
\end{proof}

Чтобы $P (\exists i:$листовой робот попал в $L_i(x_j)$ в состоянии $q)>\varepsilon$ надо, чтобы $\mathlarger\sum_{i}P($листовой робот попал в $L_i(x_j)$ в состоянии $q)>\varepsilon$, но тогда верно, что $\exists i:P($листовой робот попал в $L_i(x_j)$ в состоянии $q)>\frac{\varepsilon}{l}$. А множество таких $x_j$ конечно. Значит клеток, которые гарантированно посетит робот с вероятностью хотя бы $\varepsilon$ не больше конечного объединения подпространств четвёртой степени.
Проведя аналогичные рассуждения для других состояний этой листовой компоненты получим, что множество клеток, в которые попадает листовой робот с вероятностью хотя бы $m\cdot\varepsilon$ конечно. Тогда взяв $\varepsilon<\frac{1}{m}$ получим, что листовой робот гарантировано побывает в множестве клеток, покрываемым конечным объединением подпространств четвёртой степени, то есть изначальный робот не может обойти $\zs^5$
\end{paragraph}
	
\end{subsubsection}
\end{subsection}

\begin{subsection}{Робот с генератором случайных чисел, камнем и флажком.}

Перейдём к случаю робота с камнем и флажком. Изначальное положение камня и робота в клетке с нулевыми координатами, там же находится единственный элемент из $M$.

\begin{subsubsection}{Обход $\zs^6$}
	Для начала опишем, как, имея флажок и камень, побывать камнем на $\zs^4$, Робот ходит от камня без него по случайным векторам из множества $(\vec{e_1},\ -\vec{e_1},\ \vec{e_2},\ -\vec{e_2})$. Если он возвращается к камню не попав в процессе на клетку с флажком, то он перемещается с камнем по случайному вектору из $(\vec{e_3},\ -\vec{e_3},\ \vec{e_4},\ -\vec{e_4})$, иначе он перемещается с камнем по случайному вектору из $(\vec{e_1},\ -\vec{e_1},\ \vec{e_2},\ -\vec{e_2})$, а потом вместе с камнем переходит по случайному вектору из $(\vec{e_3},\ -\vec{e_3},\ \vec{e_4},\ -\vec{e_4})$. Из-за того, что случайное блуждание возвратно и обходит всю плоскость, робот с камнем гарантировано вернётся на плоскость $\vec{e_1},\ \vec{e_2}$. Так как он окажется там бесконечное число раз, то с вероятностью один робот в какой-то момент дойдёт до флага и обратно. Тогда камень гарантированно побывает во всех клетках плоскости $\vec{e_1},\ \vec{e_2}$ бесконечное число раз и из каждой клетки плоскости проблуждает вдоль $\vec{e_3},\ \vec{e_4}$.
	Добавив случайное блуждание от камня и обратно по векторам из множества $(\vec{e_5},\ -\vec{e_5},\ \vec{e_6},\ -\vec{e_6})$ после каждого перемещения камня получим обход $\zs^6$.
	
\end{subsubsection}

\begin{subsubsection}{Невозможность обхода $\zs^7$}
	Рассмотрим блуждания робота. Они могут иметь следующие виды:
\begin{itemize}
	\item Перемещение с камнем,
	\item Перемещение от камня и обратно,
	\item Перемещение от флага до камня или от камня до флага.
	\item Перемещение от флага до флага; этот вид можно не рассматривать, как отдельный тип, т.к. оно может обойти только множество клеток, являющихся подмножеством конечного числа плоскостей.
\end{itemize}
	
Заметим, что чтобы обходить $\zs^7$ мы всегда должны с вероятностью один возвращаться к флагу, иначе мы и $\zs^5$ не сможем обойти, и с вероятностью один возвращаться к камню, так как флаг --- это камень, который мы не можем таскать, а робот с камнем не может обойти даже $\zs^5$.

	Посмотрим на множество клеток, где может располагаться камень, чтобы от него можно было дойти до флага.

	Оно совпадает с множеством клеток до которых можно дойти от камня, а это множество является подмножеством объединения конечного числа плоскостей(из-за того, что блуждание робота без камня возвратное).

	Обозначим его через $L=\mathlarger\bigcup_{i}L_i$, где $L_i$-плоскости покрывающие перемещения робота от камня до камня или флага, мощность этого множества обозначим за $l$.
	Так как мы возвращаемся к флагу, то камень в какой-то момент должен возвращаться в $L({\{0\}}^k)$.
Посмотрим, как себя может вести робот с камнем между встречами с флажком.
Пусть мы в первый раз подойдём к камню в состоянии $q_1$.
Тогда покажем, какой робот без камня соответствует перемещениям робота с камнем с момента первой встречи камня после флага, до последнего отхода перед флагом.
Недетерминированный конечный автомат $R'(q'_1)$:

\begin{itemize}
	\item 	Состояния аналогичны $R$,
	\item 	Начальное состояние соответствует состоянию  $q_1$,
	\item Вероятностям перехода между состояниями, соответствующими $R$, совпадают с реализуемыми вероятностями перехода между изначальными состояния с камнем до следующего состояния с камнем. В случае, если при этом камень оставался на месте, к состояниям $R'$ добавляется дополнительные для прохода с камнем вверх и вниз (вероятность посещения клеток новым роботом от этого может только увеличиться),
	\item Обозначим количество состояний в $R'$ через $m$.
\end{itemize}

Новый робот имеет те же вероятности положения в пространстве с камнем, начиная с отхода от камня в состояние $q'_1$ до возвращения к нему.
Кроме того для клеток, где может оказаться $R$ с камнем верно, что блуждая из этой клетки $R'$ попадёт в множество клеток $L({\{0\}}^7)$ с вероятностью 1.
Покажем, что тогда множество клеток, которые он может посетить является конечным объединением подпространств размерности четыре.
Рассмотрим его ориентированный граф компонент сильной связности.
Возьмём какое-нибудь из состояние $q'$.м
Этот факт очевиден для случая, когда таких клеток конечно, так их можно покрыть конечным количеством подпространств, так что будем рассматривать только $q'$ для которых существует сколь угодно далёкая клетка с вероятностью её посетить в состояние $q'$.

Докажем это сначала для случая, где $q'$ находится в листовой компоненте.
Возьмём случайное блуждание, построенное на переходах в пространстве робота между двумя состояниями $q'$ (между ними могут быть только состояния отличные от $q'$).

Обозначим за $p$ минимальную вероятность дойти из состояния того же листа $q'_2$ до $q'$ за не более, чем $m$ шагов.
Так как на блуждание не наложены ограничения остановки, то из любой клетки, где $R'$ может побывать в состояние $q'$(обозначим их через $B$) можно с вероятностью хотя бы $p$ попасть в клетку из $L'=\mathlarger\bigcup_{d(\vec{x},{\{0\}}^7)\leq m} L(\vec{x})$ в состояние $q'$.
Возьмём конечное множество подпространств размерности четыре $L'_i: L'=\mathlarger\bigcup_{i} L'_i$, их количество $l'$.
Обозначим за $B_i\subset B:$ блуждая из них можно с вероятностью хотя бы $\frac{p}{l'}$ попасть в клетку из $L'_i$.
Тогда $B=\mathlarger\bigcup_{i=1}^{l'} B_i$.

Докажем, что множество клеток, из которых случайным блужданием можно попасть в $L'_i$ с вероятностью хотя бы $\frac{p}{l'}$ является подмножеством конечного объединения подпространств размерности четыре, тогда будет верно и то, что $B$ является подмножеством конечного объединения подпространств размерности четыре.
	Для этого введём на $\zs^7$ отношение эквивалентности $\vec{x_1}\sim\vec{x_2}$, которое верно при $L'_i(\vec{x_1})=L'_i(\vec{x_2})$.
	Посмотрим, что произойдёт с блужданием на склейке.
	Обозначим новый лабиринт через $Z'$, а его случайное блуждание через $X'$.
	Если $X'$ возвратное, то достижимые клетки задаются не более чем двумя векторами, но тогда $L'_i(B_i)$ уже подпространство размерности четыре.
	Из невозвратности $X'$ следует, что в $Z'$ конечное число клеток с вероятностью попадания хотя бы $\frac{p}{l'}$.
	Развернув отображение обратно получим, что множество клеток $\zs^7$, в которые робот попадёт при блуждание от камня соответствущее плоскости $L'_i$ с вероятностью хотя бы $\frac{p}{l'}$ будет подмножеством конечного объединения плоскостей.
	Так как вероятность робота попасть способом, соответствующим $L'_i$, в прообраз не превосходит вероятности из образа.

\begin{remark}
	Если компонента с состоянием $q$ не листовая, то из неё за не более, чем $m$ шагов можно бы было дойти до состояния $q'$. Тогда клетки, где мы можем оказаться в состояние $q$ находятся на расстояние не больше $m$ от подмножества клеток объединения конечного числа подпространств размерности четыре, что тоже является подмножеством клеток объединения конечного числа подпространств размерности четыре.
\end{remark}

Количество разных $R'$ конечно, в каждом из них конечное число состояний для каждого из которых верно, что множество клеток в котором он побывает является подмножеством объединения конечного числа подпространств размерности четыре. Значит робот $R$ побывает с камнем в множестве клеток, являющемся подмножеством объединения конечного числа подпространств размерности четыре. Но так как от камня робот уходит не дальше $L$, то множество клеток, где он побывает является подмножеством объединения конечного числа подпространств размерности шесть. Но тогда он не обходит $\zs^7$.

\end{subsubsection}

\end{subsection}

\begin{subsection}{Робот с генератором случайных чисел, камнем и плоскостью флажком}

Осталось рассмотреть случай робота с камнем и плоскостью флажков. Изначальное положение камня и робота  в клетке с нулевыми координатами. Из неё же выходит плоскость флажков с координатами $(0,0,0,0,0,0,a,b)$.

\begin{subsubsection}{Обход $\zs^8$}
\end{subsubsection}

Чтобы построить обход $\zs^8$ просто слегка модернизируем программу робота, обходящего $\zs^8$ с камнем и флажком. Оказавшись на флаге робот должен сделать случайный, равновероятный ход по одному из векторов $(\vec{e_7},\ \vec{-e_7},\ \vec{e_8},\ \vec{-e_8}, 0)$. Переход на $\vec{0}$ делаем ходом по вектору $\vec{e_7}$, а потом обратно. Из-за того, что это блуждание побывает во всех клетках плоскости бесконечное число раз, а изначальный робот был во всех клетках  $\zs^6$ (и находился на клетке с флагом и камнем бесконечное число раз), то теперь робот побывает в каждой клетке пространстве $\zs^8$,

\begin{subsubsection}{Невозможность обхода $\zs^9$}
	Этот случай не требует долгих доказательств, так как может быть доказан по аналогии со случем одного флажка. Множество клеток, где может оказаться камень будет являться объединением конечного числа ``подпространств''\  размерности 6. То есть клетки, которые гарантировано посетит робот, являются объединением конечного числа "подпространств"\  размерности 8. Значит робот не обойдёт $\zs^9$.
\end{subsubsection}

\end{subsection}

\end{document}